\documentclass[11pt]{article}
\usepackage{graphicx}    
\usepackage{liyang}

\def\colorful{0}

\ifnum\colorful=1

\newcommand{\blue}[1]{{{\color{blue}#1}}}
\newcommand{\red}[1]{{\color{red} {#1}}}

\fi
\ifnum\colorful=0

\newcommand{\blue}[1]{{{#1}}}
\newcommand{\red}[1]{{{#1}}}

\fi

\def\Dyes{\mathcal{D}_{\text{yes}}}
\def\Dno{\mathcal{D}_{\text{no}}}

\def\obM{\overline{\bM}}
\def\oM{\overline{M}}

\def\py{p_{\text{yes}}}
\def\pn{p_{\text{no}}}

\def\SSSQ{\textsf{SSSQ}}
\def\SSEQ{\textsf{SSEQ}}

\def\Bin{\mathrm{Bin}}
\def\Alg{\mathrm{Alg}}

\def\eps{\epsilon}

\title{Settling the query complexity of non-adaptive junta testing}
\author{
Xi Chen\thanks{Columbia University, email: \texttt{xichen@cs.columbia.edu}.} 
\and
Rocco A. Servedio\thanks{Columbia University, email: \texttt{rocco@cs.columbia.edu}.}
\and
Li-Yang Tan\thanks{Toyota Technological Institute, email: \texttt{liyang@cs.columbia.edu}.} 
\and
Erik Waingarten\thanks{Columbia University, email: \texttt{eaw@cs.columbia.edu}.}
\and
Jinyu Xie\thanks{Columbia University, email: \texttt{jinyu@cs.columbia.edu}}}

\begin{document}

\def\Ayes{\mathcal{A}_{\text{yes}}}
\def\Ano{\mathcal{A}_{\text{no}}}
\def\py{p}
\def\pn{q}
\def\Alg{\textsc{Alg}}
\def\Vyes{\mathcal{V}_{\text{yes}}}
\def\Vno{\mathcal{V}_{\text{no}}}
\def\Byes{\mathcal{B}_{\text{yes}}}
\def\Bno{\mathcal{B}_{\text{no}}}

\maketitle

\begin{abstract}
We prove that any non-adaptive algorithm that tests whether an unknown 
  Boolean function $f\colon \zo^n\to\zo $ is a $k$-junta or $\eps$-far from every $k$-junta must make $\smash{\widetilde{\Omega}(k^{3/2} / \eps)}$ many queries for a wide range of parameters $k$ and $\eps$. Our result dramatically improves previous lower bounds from \cite{BGSMdW13, STW15}, and is essentially optimal given Blais's non-adaptive junta tester from \cite{Blais08}, which makes  $\widetilde{O}(k^{3/2})/\eps$ queries.  Combined with the adaptive tester of \cite{Blaisstoc09} which makes $O(k\log k + k /\eps)$ queries, our result shows that adaptivity enables polynomial savings in query complexity for junta testing.
  \end{abstract}

\thispagestyle{empty}

\newpage

\setcounter{page}{1}


\section{Introduction}
This paper is concerned with the power of adaptivity in property testing, specifically property testing of Boolean functions.  At a high level, a property tester for Boolean functions is a randomized algorithm which, given black-box query access to an unknown and arbitrary Boolean function $f\colon \zo^n \to \zo$, aims to distinguish between the case that $f$ has some particular property of interest versus the case that $f$ is far in Hamming distance from every Boolean function satisfying the property.  The main goals in the study of property testing algorithms are to develop testers that make \emph{as few queries} as possible, and to establish lower bounds matching these query-efficient algorithms.  Property testing has by
now been studied for many different types of Boolean functions, including linear functions and low-degree polynomials over $GF(2)$ \cite{BLR93,AKKLRtit,BKSSZ10}, literals, conjunctions, $s$-term monotone and non-monotone DNFs \cite{PRS02,DLM+:07}, monotone and unate functions \cite{GGLRS,FLNRRS,CS13a,CST14,CDST15,KMS15,BB15,KS16,CS16,BMPR16}, various types of linear threshold functions \cite{MORS:10,MORS:09random,BlaisBM11}, size-$s$ decision trees and $s$-sparse $GF(2)$ polynomials and parities \cite{DLM+:07,BlaisBM11,BlaisKane12}, functions with sparse or low-degree Fourier spectrum \cite{GOS+11},
and much more.  See e.g.~\cite{Ron:08testlearn,Ron:10FNTTCS,PropertyTestingICS} for some fairly recent broad overviews
of property testing research.

In this work we consider the property of being a \emph{$k$-junta}, which is one of the earliest and most intensively studied properties in the Boolean function property testing literature.  Recall that $f$ is a $k$-junta if it has at most $k$ relevant variables, i.e.,~there exist $k$ distinct indices $i_1,\dots,i_{k}$ and a $k$-variable function $g\colon \{0,1\}^{k}
\to \{0,1\}$ such that $f(x) = g(x_{i_1},\dots,x_{i_{k}})$ for all $x \in \{0,1\}^n$.
Given $k=k(n)\colon\mathbb{N}\rightarrow \mathbb{N}$ and   $\eps=\eps(n)\colon
  \mathbb{N}\rightarrow \mathbb{R}_{>0}$,
we say an algorithm which has black-box access to an unknown and arbitrary $f\colon \zo^n \to \zo$ is an \emph{$\eps$-tester} or \emph{$\eps$-testing algorithm for $k$-juntas}
  if
  it accepts with probability at least $5/6$ when $f$ is a $k(n)$-junta and rejects with probability at least $5/6$ when
  $f$ is $\eps(n)$-far from all $k(n)$-juntas (meaning that $f$ disagrees with any $k(n)$-junta $g$ on at least
 $\eps(n) \cdot 2^n$ many inputs).

\ignore{
A testing algorithm for $k$-juntas is given as input $k$ and $\eps > 0$, and is provided with black-box
oracle access to an unknown and arbitrary $f\colon \{0,1\}^n \to \{0,1\}$.  The algorithm must output
``yes'' with high probability (say at least 5/6) if $f$ is a $k$-junta, and must output ``no''
with high probability if $f$ disagrees with every $k$-junta on at least an $\eps$
fraction of all possible inputs.
}

Property testers come in two flavors, adaptive and non-adaptive.  An adaptive tester receives the value of $f$ on its $i$-th query string before deciding on its $(i+1)$-st query string, while a non-adaptive tester selects all of its query strings before receiving the value of $f$ on any of them.  Note that non-adaptive testers can evaluate all of their queries in one parallel stage of execution, while this is in general not possible for adaptive testers.  This means that if evaluating a query is very time-consuming, non-adaptive algorithms may sometimes be preferable to adaptive algorithms even if they require more queries. For this and other reasons, it is of interest to understand when, and to what extent, adaptive algorithms can use fewer queries than non-adaptive algorithms (see \cite{RonTsur11,RonServedio:13} for examples of property testing problems where indeed adaptive algorithms are provably more query-efficient than non-adaptive ones).\ignore{abilities and limitations of non-adaptive testing algorithms.}
\ignore{\enote{Can we cite a paper where being non-adaptive was necessary? {\bf Rocco:}  Added cites.  I think we shouldn't get into the actual results in these papers -- both are for problems where adaptive algorithms use exponentially fewer queries than non-adaptive algorithms, but if we highlight this it may implicitly attention to the fact that here for juntas we already knew the gap was at most $k$ versus $k^{3/2}$. We could also cite our own upcoming glorious ICALP submission but maybe it's best not to both b/c it's a little tacky to cite one's own unpublished work, and because it's for a restricted testing problem where we require the LTF guarantee.}
}

The query complexity of adaptive junta testing algorithms is at this point well understood.  In \cite{ChocklerGutfreund:04}
Chockler and Gutfreund showed that even adaptive testers require $\Omega(k)$ queries to distinguish $k$-juntas from random functions on $k+1$ variables, which are easily seen to be constant-far from $k$-juntas.   Blais  \cite{Blaisstoc09}
gave an adaptive junta testing algorithm that uses only $O(k \log k + k/\eps)$ queries, which is optimal (for constant $\eps$) up to a  multiplicative factor of $O(\log k)$.

Prior to the current work, the picture was significantly less clear for non-adaptive junta testing.  In the first work on junta testing, Fischer et al.~\cite{FKRSS03} gave a non-adaptive tester that makes $O(k^2 (\log k)^2/\eps)$ queries.  This was improved by Blais
\cite{Blais08}
with a non-adaptive tester that uses only $O(k^{3/2} (\log k)^3/\eps)$ queries.   On the lower bounds side, \cite{Blais08} also showed that for all $\eps \geq k/2^k$, any non-adaptive algorithm for $\eps$-testing $k$-juntas must make $\Omega\left( k/ ({\eps \log(k/\eps)})\right)$ queries. Buhrman~et al. \cite{BGSMdW13} gave an $\Omega(k \log k)$ lower bound (for constant $\eps$) for non-adaptively testing whether a function $f$ is a size-$k$ parity; their argument also yields an $\Omega(k \log k)$ lower bound (for constant~$\eps$) for non-adaptively $\eps$-testing $k$-juntas.
More recently, \cite{STW15} obtained a new lower bound for~non-adaptive junta testing that is incomparable to both the \cite{Blais08} and the \cite{BGSMdW13} lower bounds.  They showed that for all $\eps:k^{-o_k(1)} \le \eps \le o_k(1)$, any non-adaptive $\eps$-tester for $k$-juntas must make
\[
\Omega\left({\frac {k \log k}
{\eps^{c} \log( \log(k)/\eps^{c})}}\right)
\]
many queries, where $c$ is any absolute constant less than 1.  For certain restricted values of $\eps$ such as $\eps=1/\log k$, this lower bound is larger than the $O(k/\eps + k \log k)$ upper bound for \cite{Blaisstoc09}'s adaptive algorithm, so the \cite{STW15} lower bound shows that in some restricted settings, adaptive junta testers can outperform non-adaptive ones.  However, the difference in performance is quite small, at most a $o(\log k)$ factor.  We further note that all of the lower bounds \cite{Blais08,BGSMdW13,STW15} are of the form $\widetilde{\Omega}(k)$ for constant $\eps$, and hence rather far from the $\widetilde{O}(k^{3/2})/\eps$ upper bound of \cite{Blais08}.

\ignore{
\blue{

We motivate our work by observing that juntas are a very basic type of Boolean function whose study intersects many different areas within theoretical computer science.  In complexity theory and cryptography, $k=O(1)$-juntas are precisely the Boolean functions computed by $\mathsf{NC^0}$ circuits.  Juntas arise naturally in settings where a small (unknown) set of features determines the label of a high-dimensional data point, and hence many researchers in learning theory have studied juntas across a wide range of different learning models, see e.g.~\cite{Blum:94,DhagatHellerstein:94,BlumLangley:97,GTT:99,ArpeReischuk:03, Mos:04,AticiServedio:07qip,ArpeMossel10,FGKP:journal,gregvaliantfocs12,DSFTWW15}.
Finally, the problem of testing whether
an unknown Boolean function is a $k$-junta is one of the most thoroughly studied questions in
Boolean function property testing.  We briefly survey relevant previous work on testing juntas in the following subsection.

\subsection{Prior work on testing juntas}

Fischer et al.~\cite{FKRSS03} were the first to explicitly consider the junta testing problem.  Their
influential paper gave several algorithms for testing $k$-juntas, the most efficient of which is
a non-adaptive tester that makes $O(k^2 (\log k)^2/\eps)$ queries.  This was improved by Blais
\cite{Blais08} who gave a non-adaptive testing algorithm that uses only $O(k^{3/2} (\log k)^3/\eps)$ queries;
this result is still the most efficient known non-adaptive junta tester.  Soon thereafter Blais \cite{Blaisstoc09}
gave an \emph{adaptive} junta testing algorithm that uses only $O(k \log k +
k/\eps)$ queries, which remains the most efficient known junta testing
algorithm to date.

We note that ideas and and techniques from these junta testing algorithms have played an important role in a broad range of algorithmic results for
other Boolean function property testing problems.    These include efficient algorithms for testing various classes of functions, such as $s$-term DNF
formulas, small Boolean circuits, and sparse $GF(2)$ polynomials, that are close to juntas but
not actually juntas themselves (see e.g.~\cite{DLM+:07,GOS+11,DLMW10:algorithmica,chakraborty2011efficient}),  as well as algorithms for testing linear threshold functions \cite{MORS:10} (which in general are not close to juntas).
Junta testing is also closely related to the problem of Boolean function isomorphism testing, see e.g.~\cite{BO10, BWY12, ChakrabortyFGM12, AlonBCGM13}.

Lower bounds for testing $k$-juntas have also been intensively studied.  The original
\cite{FKRSS03} paper gave an $\Omega(\sqrt{k}/\log k)$ lower bound for nonadaptive algorithms that test whether
an unknown function is a $k$-junta versus constant-far from every $k$-junta.
Chockler and Gutfreund \cite{ChocklerGutfreund:04}
simplified, strengthened and extended this lower bound by proving that even \emph{adaptive} testers require
$\Omega(k)$ queries to distinguish $k$-juntas from random functions on $k+1$ variables, which are easily seen to be constant-far
from $k$-juntas.   (We describe the construction and sketch
the \cite{ChocklerGutfreund:04} argument in Section \ref{sec:techniques}
below).  Blais~\cite{Blais08} was the first to give a lower bound that
involves the distance parameter $\eps$; he showed that for $\eps \geq k/2^k$,
any non-adaptive algorithm for $\eps$-testing $k$-juntas must
make $\Omega\left( {\frac k {\eps \log(k/\eps)}}\right)$ queries.

In recent years numerous other works have given junta testing lower bounds.  In \cite{BlaisBM11} Blais, Brody and Matulef established a connection between lower bounds in communication complexity and property testing lower bounds, and used this connection
(together with known lower bounds on the communication complexity of the size-$k$ set disjointness problem) to give a different proof of an $\Omega(k)$ lower bound for adaptively testing whether a function is a $k$-junta versus constant-far from every $k$-junta.  More recently,
Blais, Brody and Ghazi \cite{BlaisBG14} gave new bounds on the communication complexity of the Hamming distance function, and used these bounds to give an alternate proof of the $\Omega(k)$ lower bound for adaptive junta testing algorithms via the \cite{BlaisBM11} connection.
Blais and Kane  \cite{BlaisKane12} studied the problem of testing whether an $n$-variable Boolean function is a size-$k$ parity function (as noted in \cite{BlaisKane12}, lower bounds for this problem give lower bounds for testing juntas), and via a geometric and Fourier-based analysis gave a $k-o(k)$ lower bound for adaptive algorithms and a $2k-O(1)$ lower bound for non-adaptive algorithms, again for $\eps$ constant.
Buhrman et al.~\cite{BGSMdW13} combined the communication complexity based approach of
\cite{BlaisBM11} with an $\Omega(k \log k)$ lower bound for the one-way communication complexity of
$k$-disjointness to obtain an $\Omega(k \log k)$ lower bound (for constant $\eps$) for testing whether a function $f$ is a size-$k$ parity, and hence for testing whether $f$ is a $k$-junta.

\subsection{Our main result:  Adaptivity helps for testing juntas}

While the junta testing problem has been intensively studied, the results
described above still leave a gap between the query complexity of the best
\emph{adaptive} algorithm \cite{Blaisstoc09} and the strongest known lower
bounds for \emph{non-adaptive} junta testing.  The lower bounds of
$\Omega\left( {\frac k {\eps \log(k/\eps)}}\right)$ from \cite{Blais08} and
$\Omega(k \log k)$ (for $\eps$ constant) from \cite{BGSMdW13} are incomparable, but neither of them is strong enough, for any setting of
$\eps$, to exceed the $O(k \log k + k/\eps)$ upper bound from
\cite{Blaisstoc09}.  In \cite{Blais08} Blais asked as an open question
``\emph{Is there a gap between the query complexity of adaptive and non-adaptive algorithms for
testing juntas?}''  This question was reiterated in a 2010 survey article on testing juntas,
in which Blais explicitly asked   ``\emph{Does adaptivity help when testing
$k$-juntas?}'', referring to this as a ``basic problem''~\cite{Blais10survey}.

Our main contribution in the present work is to give a better lower bound on non-adaptive junta testing algorithms
which implies that the answer to the above questions is ``yes.''  We prove
the following:

\begin{theorem} \label{thm:main}
Let $\calA$ be any non-adaptive algorithm which tests whether an unknown black-box $f\colon
\{0,1\}^n \to \{0,1\}$ is a $k$-junta versus $\eps$-far from every $k$-junta. Then for all $\eps $ satisfying $k^{-o_k(1)} \le \eps \le o_k(1)$, algorithm $\calA$ must make at least
\begin{equation} \label{eq:qbound}
q =
{\frac {C k \log k}
{\eps^{c} \log( \log(k)/\eps^{c})}}
\end{equation}
queries, where $c$ is any absolute constant $< 1$ and $C>0$ is an absolute constant.
\end{theorem}

For suitable choices of $\eps$, such as $\eps = 1/(\log k)$, the
lower bound of Theorem~\ref{thm:main} is asymptotically larger than the
$O(k \log k + k/\eps)$ upper bound of the \cite{Blaisstoc09} adaptive
algorithm.  Thus, together with the
\cite{Blaisstoc09} upper bound, our lower bound gives an affirmative answer
to the question
posed in \cite{Blais08,Blais10survey}:  adaptivity helps for testing $k$-juntas.\footnote{We  note in this context that several other natural Boolean function classes are known to exhibit a gap between the query complexity of adaptive versus non-adaptive testing algorithms.  These include the class of signed majority functions
\cite{MORS:09random,RonServedio:13}
and the class of read-once width-two OBDD \cite{RonTsur-OBDD}.  In all three cases  the adaptive tester which beats the best possible non-adaptive tester may be viewed as performing some sort of binary search.}

It is interesting that while all of the recent junta testing lower bounds \cite{BlaisBM11, BlaisBG14, BGSMdW13} employ the connection with communication complexity lower bounds that was established in \cite{BlaisBM11}, our proof of
Theorem \ref{thm:main} does not follow this approach.  Instead, we give a proof using Yao's classic minimax principle; however,
our argument is somewhat involved, employing a new Boolean isoperimetric inequality and a very delicate application of a variant of McDiarmid's ``method of bounded differences'' that allows for a (low-probability) bad event.
In the rest of this section we motivate and explain our approach at a high level before giving the full proof in the subsequent sections.
}

}

\subsection{Our results}

The main result of the paper is the following theorem:
\begin{theorem}
\label{thm:main-intro}
Let $\alpha\in (0.5,1)$ be an absolute constant.
Let $k=k(n)\colon\mathbb{N}\rightarrow \mathbb{N}$ and $\eps=\eps(n)\colon
  \mathbb{N}\rightarrow \mathbb{R}_{>0}$ be two functions that satisfy
$ k(n)\le \alpha n$ and $2^{-n}\le \eps(n)\le 1/6$ for all sufficiently large $n$.
Then any non-adaptive $\eps$-tester for $k$-juntas must make
  $\widetilde{\Omega}(k^{3/2} / \eps)$ many queries.
\end{theorem}

Together with the $\widetilde{O}(k^{3/2})/\eps$ non-adaptive upper bound from \cite{Blais08}, Theorem~\ref{thm:main-intro} settles the query complexity of non-adaptive junta testing up to poly-logarithmic factors.

\subsection{High-level overview of our approach}

Our lower bound approach differs significantly from previous work.  Buhrman et al.~\cite{BGSMdW13} leveraged the connection between communication complexity lower bounds and property testing lower bounds that was established in the work of
\cite{BlaisBM11} and applied an $\Omega(k \log k)$ lower bound on the one-way communication complexity of $k$-disjointness to establish their lower bound.  Both \cite{Blais08} and \cite{STW15} are based on edge-isoperimetry results for the Boolean hypercube (the edge-isoperimetric inequality of
Harper \cite{Harper64},
Bernstein \cite{Bernstein67},
Lindsey~\cite{Lin64},
and Hart \cite{Hart76} in the case of \cite{Blais08}, and a slight extension of a result of Frankl~\cite{Frankl83} in \cite{STW15}).  In contrast, our lower bound argument takes a very different approach; it consists of a sequence of careful reductions, and employs an \emph{upper} bound on the total variation distance between two Binomial distributions (see Claim~\ref{lem:dtvbound}).


Below we provide a high level overview of the proof of the lower bound given by Theorem~\ref{thm:main-intro}. First, it is not difficult to show that Theorem~\ref{thm:main-intro} is a consequence of the following more specific lower bound for the case where  $k = \alpha n$:

\begin{theorem}
\label{thm:main}
Let $\alpha\in (0.5,1)$ be an absolute constant.
Let $k=k(n)\colon\mathbb{N}\rightarrow \mathbb{N}$ and $\eps=\eps(n)\colon
  \mathbb{N}\rightarrow \mathbb{R}_{>0}$ be two functions that satisfy $k(n)=\alpha n$ and
$ 2^{-(2\alpha-1)n/2}\le \eps(n)\le {1/6}
$
for sufficiently large $n$.
Then any non-adaptive $\eps$-tester for $k$-juntas
  must make
  $\widetilde{\Omega}(n^{3/2}/\eps)$ many queries.
\end{theorem}

See Appendix~\ref{app:1} for the proof that Theorem~\ref{thm:main} implies Theorem~\ref{thm:main-intro}.

We now provide a sketch of how Theorem~\ref{thm:main} is proved.\ignore{The constant $\alpha \in (0.5,1)$ in the statement of Theorem~\ref{thm:main} is
a technical necessity for the lower bounds, though it will} It may be convenient for the reader,
  on the first reading, to consider $\alpha = 3/4$ and to think of $\eps$ as being a small constant such as $0.01$.

Fix a sufficiently large $n$. Let
  $k=\alpha n$ and $\eps=\eps(n)$ with $\eps$ satisfying the
  condition in Theorem \ref{thm:main}.
We proceed by Yao's principle and prove lower bounds for deterministic
  non-adaptive algorithms which receive inputs drawn from one of two probability distributions, $\Dyes$ and $\Dno$, over $n$-variable Boolean functions.  The distributions $\Dyes$ and $\Dno$ are designed so that a Boolean function $\boldf \leftarrow \Dyes$ is a $k$-junta with probability $1 - o(1)$ and $\boldf \leftarrow \Dno$ is $\eps$-far from every $k$-junta with probability $1 - o(1)$. In Section~\ref{sec:distributions} we define $\Dyes$ and~$\Dno$, and establish the above properties.
By Yao's principle, it then suffices to show that any $q$-query non-adaptive deterministic algorithm (i.e., any set of $q$ queries) that succeeds in distinguishing them must have $q=\widetilde{\Omega}(n^{3/2}/\eps)$.

This lower bound proof consists of two components:
\begin{flushleft}\begin{enumerate}
\item A reduction from a simple algorithmic task called Set-Size-Set-Queries (\SSSQ\
 for short), which we discuss informally later in this subsection and we define formally in Section~\ref{sec:simple-task}.  This reduction implies that the non-adaptive deterministic query complexity of distinguishing
  $\Dyes$ and $\Dno$ is at least as large as that of \SSSQ.
\item A lower bound of $\widetilde{\Omega}(n^{3/2}/\eps)$
  for the query complexity of \SSSQ.
\end{enumerate}\end{flushleft}

Having outlined the formal structure of our proof, let us give some intuition which may hopefully be helpful in motivating our construction and reduction.  Our yes-functions and no-functions have very similar structure to each other, but are constructed with slightly different parameter settings. The first step in drawing a random function from $\Dyes$ is choosing a uniform random subset $\bM$ of $\Theta(n)$ ``addressing'' variables from $x_1,\dots,x_n$.  A random subset $\bA$ of the complementary variables $\overline{\bM}$ is also selected, and for each assignment to the variables in $\bM$ (let us denote such an assignment by $i$), there is an independent random function $\bh_i$ over a randomly selected subset $\bS_i$ of the variables in $\bA.$  A random function from $\Dno$ is constructed in the same way, except that now the random subset $\bA$ is chosen to be slightly larger than in the yes-case.  This disparity in the size of $\bA$ between the two cases causes random functions from $\Dyes$ to almost always be $k$-juntas and random functions from $\Dno$ to almost always be far from $k$-juntas.

An intuitive explanation of why this construction is amenable to a lower bound for non-adaptive algorithms is as follows.  Intuitively, for an algorithm to determine that it is interacting with (say) a random no-function rather than a random yes-function, it must determine that the subset $\bA$ is larger than it should be in the yes-case.  Since the set $\bM$ of $\Theta(n)$ many ``addressing'' variables~is selected randomly, 
if a non-adaptive algorithm uses two query strings $x,x'$ that differ in more~than a few coordinates, it is very likely that they will correspond to two different random functions $\bh_i,\bh_{i'}.$
Hence every pair of query strings $x,x'$ that correspond to the same  $\bh_i$ 
  can differ only in~a few coordinates in $\overline{\bM}$, with high probability, which significantly limits the
  power of a non-adaptive algorithm in distinguishing $\Dyes$ and $\Dno$ no matter which set of  query strings it picks.
   This makes it possible for us to reduce from the \SSSQ~problem to the problem of distinguishing 
   $\Dyes$ and $\Dno$ at the price of only a small quantitative cost in query complexity, see Section~\ref{sec:reduction}.

At a high level, the $\SSSQ$ task involves distinguishing\ignore{\rnote{Was ``whether
  a sequence of $n$ random bits is biased or not.
The random bits are used to define a hidden set which the algorithm has (a limited form~of) query access to.'' But it seems weird to say that the random bits define the hidden set, so I changed it to the stuff in red.}}
whether or not a hidden set (corresponding to $\bA$)  is ``large.''  An algorithm for this task can only access certain random bits, whose biases are determined by the hidden set and whose exact distribution is inspired by the exact definition of the random functions $\bh_i$ over the random subsets $\bS_i$. Although \SSSQ ~is an artificial problem, it is much easier to work with compared~to the original problem of distinguishing  $\Dyes$ and $\Dno$.  \red{In particular, we give a reduction from an even simpler algorithmic task called Set-Size-Element-Queries (\SSEQ\ for short) to \SSSQ ~(see Section~\ref{task2def}) and the query complexity lower bound for \SSSQ
  ~follows directly from the lower bound for \SSEQ ~presented in Section \ref{task2proof}.
  
Let us give a high-level description of the \SSEQ~task to provide some intuition for how we prove a query lower bound on it.
Roughly speaking, in this task an oracle holds an unknown and random subset $\bA$ of $[m]$ (here $m=\Theta(n)$) which is either ``small'' (size roughly $m/2$) or ``large'' (size roughly $m/2 + \Theta(\sqrt{n} \cdot \log n)$), and the task is to determine whether $\bA$ is small or large.  The algorithm may repeatedly query the oracle by providing it, at the $j$-th query, with an element $i_j \in [m]$; if $i_j \notin \bA$ then the oracle responds ``0'' with probability 1, and if $i_j \in \bA$ then the oracle responds ``1'' with probability $\eps/\sqrt{n}$ and  ``0'' otherwise.  Intuitively, the only way for an algorithm to determine that the unknown set $\bA$ is (say) large, is to determine that the fraction of elements of $[m]$ that belong to $\bA$ is $1/2 +  {\Theta(\log n/\sqrt{n})}$ rather than $1/2$; this in turn intuitively requires sampling $\Omega(n/\log^2 n)$ many random elements of $[m]$ and for each one ascertaining with high confidence whether or not it belongs to $\bA$.  But the nature of the oracle access described above for \SSEQ\ is such that for any given $i \in [m]$, at least $\Omega(\sqrt{n}/\eps)$ many repeated queries to the oracle on input $i$ are required in order to reach even a modest level of confidence as to whether or not $i \in \bA.$}  As alluded to earlier, the formal argument establishing our lower bound on the query complexity of \SSEQ~relies on an upper bound on the total variation distance between two Binomial distributions.
  
\ignore{  
  }

\subsection{Organization and Notation}

We start with the definitions of $\Dyes$ and $\Dno$ 
  as well as proofs of their properties in Section \ref{sec:distributions}.
We then introduce \SSSQ\ in Section~\ref{sec:simple-task},
  and give a reduction from \SSSQ\ to the problem of distinguishing $\Dyes$ and $\Dno$
  in Section \ref{sec:reduction}.
More formally, we show that any non-adaptive deterministic algorithm that
  distinguishes $\Dyes$ and $\Dno$ 
  can be used to solve \SSSQ\ with only an $O(\log n)$ factor loss in the query complexity.
Finally, we prove in Section~\ref{sec:proof-sssq} a lower bound for the query complexity of \SSSQ.
Theorem \ref{thm:main} then follows by combining this lower bound with  the reduction in Section~\ref{sec:reduction}.

We use boldfaced letters such as $\boldf,\bA,\bS$ to denote random variables. Given a string
   $x \in \zo^n$ and $\ell \in [n]$, we write $x^{(\ell)}$ to denote the string obtained from $x$ by flipping the $\ell$-th coordinate.
An \emph{edge} along the $\ell$th direction in $\{0,1\}^n$ is a pair $(x,y)$ of strings
  with $y=x^{(\ell)}$.
We say an edge $(x,y)$ is \emph{bichromatic with respect to a function $f$} (or simply $f$-bichromatic) if $f(x) \ne f(y).$
Given $x\in \{0,1\}^n$ and $S\subseteq [n]$, we use $x_{|S}\in \{0,1\}^S$ to denote the projection
  of $x$ on $S$.



\def\bU{\mathbf{U}}

\section{The $\Dyes$ and $\Dno$ distributions}
\label{sec:distributions}

Let $\alpha\in (0.5,1)$ be an absolute constant.
Let $n$ be a sufficiently large integer, with $k=\alpha n$, and let 
  $\eps$ be the distance parameter that satisfies
\begin{equation}\label{epsbound2}
2^{-(2\alpha-1)n/{2}}\le \eps\le 1/6.
\end{equation}
In this section we describe a pair of probability distributions $\Dyes$ and $\Dno$ supported over Boolean functions $f \colon$
$ \{0, 1\}^n \to \{0, 1\}$.
We then show that $\boldf\leftarrow\Dyes$ is a $k$-junta with probability $1 - o(1)$, and that $\boldf \leftarrow \Dno$ is $\eps$-far from being a $k$-junta with probability $1 - o(1)$.

We start with some parameters settings.
Define
\begin{align*}
\delta &\eqdef 1-\alpha\in (0,0.5), &  p&\eqdef\frac{1}{2},  & q&\eqdef \frac{1}{2} + \frac{\log n}{\sqrt{n}},\\
m &\eqdef2\delta n + \delta \sqrt{n} \log n, & t &\eqdef n-m=(2\alpha-1)n- {\delta\sqrt{n}\log n}, & N &\eqdef 2^t.
\end{align*}

A function $\boldf \leftarrow \Dyes$ is drawn according to the following randomized procedure:
\begin{flushleft}\begin{enumerate}
\item Sample a random subset $\bM \subset [n]$ of size $t$.
Let $\bGamma = \Gamma_\bM:\{0,1\}^n\rightarrow [N]$ be the function that maps $x\in \{0,1\}^n$
  to the integer encoded by $x_{|\bM}$ in binary plus one.
Note that $|\overline{\bM}|=n-t=m$.

\item Sample an $\bA \subseteq \obM$ 
  by including each element of $\obM$ in $\bA$ independently with probability $p$.
\item Sample independently a sequence of $N$ random subsets $\bS=(\bS_i\colon i\in [N])$ of $\bA$ as follows: for each $i\in [N]$,
  each element of $\bA$ is included in $\bS_i$ independently with probability ${\eps}/\sqrt{n}$.
Next we sample a sequence of $N$ functions $\bH = (\bh_i \colon i \in [N])$,
   by letting $\bh_i \colon \{0, 1\}^n \to \{0, 1\}$ be a random function over the coordinates in $\bS_i$,
   i.e., we sample an unbiased bit $\bz_i(b)$ for each string $b\in \{0,1\}^{\bS_i}$ independently and set
  $\bh_i(x)=\bz_i({x_{|\bS_i}}).$
\item Finally, $\boldf = \boldf_{\bM, \bA, \bH} \colon \{0, 1\}^n \to \{0, 1\}$ is defined using 
  $\bM,\bA$ and $\bH$ as follows:
\[ \boldf(x) = \bh_{\Gamma_{\bM}(x)}(x),\quad\text{for each $x\in \{0,1\}^n$.} \]

In words, an input $x$ is assigned the value $\boldf(x)$ as follows:  according to the coordinates of $x$ in the set $\bM$ (which intuitively should be thought of as unknown), one of the $N$ functions $\bh_i$ (each of which is, intuitively, a random function over an unknown subset $\bS_i$ of coordinates) is selected and evaluated on $x$'s coordinates in $\bS_i$.
For intuition, we note that both $\bM$ and $\obM$ will always be of size $\Theta(n)$,  
the size of $\bA$ will almost always be $\Theta(n)$, 
and for a given $i \in [N]$ the expected size of $\bS_i$ will typically be $\Theta(\eps\sqrt{n})$ (though the size of $\bS_i$ may not be as highly concßentrated 
  as the other sets when $\eps$ is tiny).
\end{enumerate}\end{flushleft}
A function $\boldf \leftarrow \Dno$ is generated using the same procedure except that $\bA$
  is a random subset of $\obM$ drawn by including each element of $\obM$ in $\bA$
  independently with probability $q$ (instead of $p$). See Figure~\ref{fig:example-x} for an example of how an input $x \in \{0, 1\}^n$ is evaluated by $\boldf \sim \Dyes$ or $\Dno$. 

\begin{figure}
\centering
\includegraphics[width=0.45\linewidth]{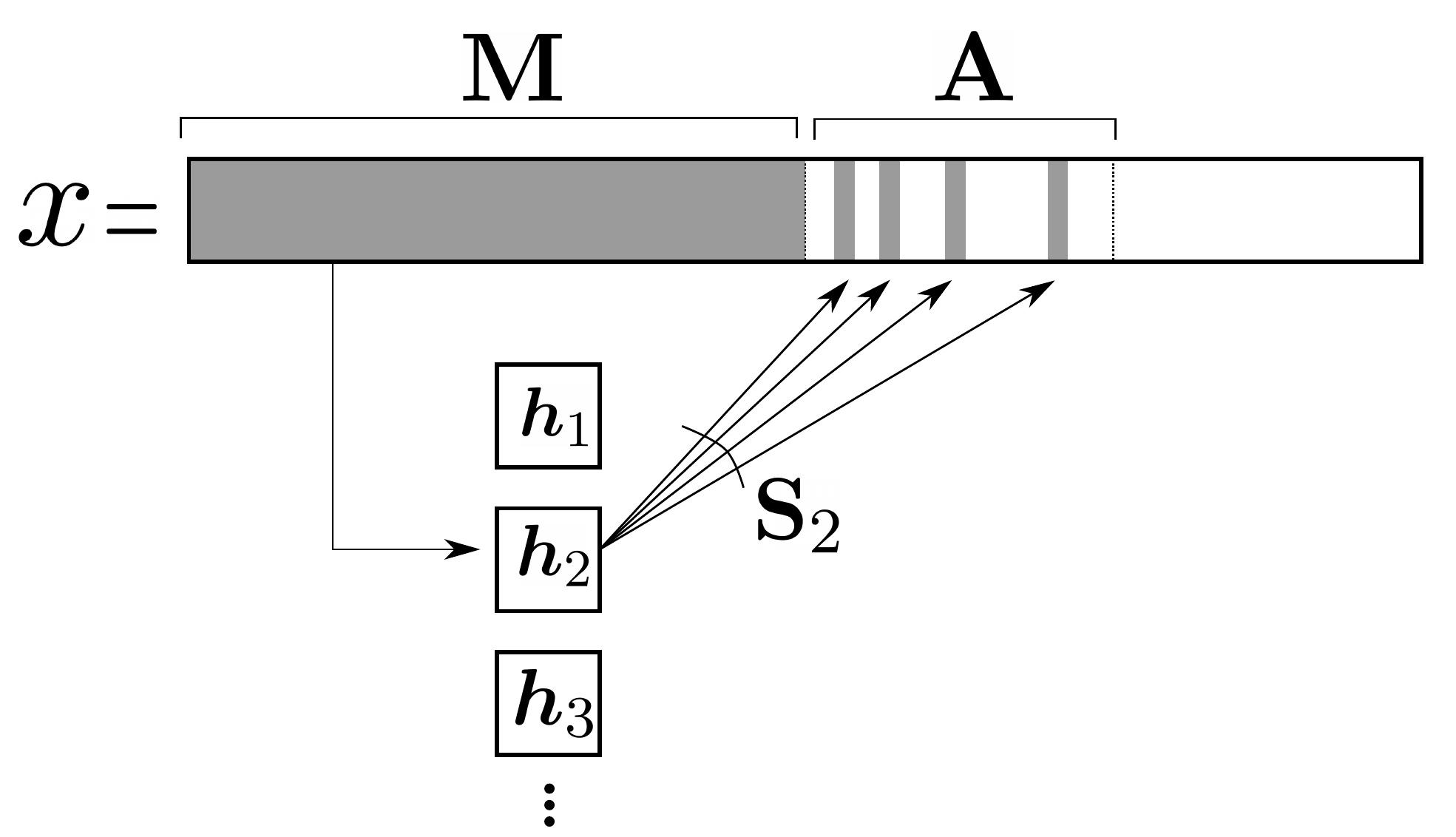}
\caption{An example of how an input $x \in \{0, 1\}^n$ is evaluated by $\boldf \sim \Dyes$ (or $\Dno$). The relevant variables of $x$ are shaded gray. All variables in $\bM$ index $x$ into $\bh_2$, which is a random function over the variables $\bS_2$, which are sampled from $\bA$ by including each with probability $ {\eps}/{\sqrt{n}}$. }
\label{fig:example-x}
\end{figure}
  

\subsection{Most functions drawn from $\Dyes$ are $k$-juntas}

We first prove that $\boldf \leftarrow\Dyes$ is a $k$-junta with probability $1-o(1)$.

\begin{lemma}
A function $\boldf \leftarrow \Dyes$ is a $k$-junta with probability $1 - o(1)$.
\end{lemma}
\begin{proof}
By the definition of $\Dyes$, all the relevant variables of 
  $\boldf\sim \Dyes$ belong to $\bM \cup \bA$. Note that $|\bM| = t$.
On the other hand, 
  the expected size of $\bA$ is $\delta n + { {\delta \sqrt{n} \log n}/ 2}$.
By a Chernoff bound,  
 $$|\bA| \leq \delta n + \frac{\delta\sqrt{n}\log n}{2}+\frac{\delta \sqrt{n} \log n}{4} < 
\delta n + \delta \sqrt{n} \log n
 $$with probability $1 - o(1)$.  When this happens we have $|\bM \cup \bA|< \alpha n = k$.
\end{proof}

\subsection{Most functions drawn from $\Dno$ are $\eps$-far from $k$-juntas}

Next we prove that $\boldf\leftarrow\Dno$ is $\eps$-far from any $k$-junta with probability $1-o(1)$.
The details of the argument are somewhat technical so we start by giving some high-level intuition, which is relatively simple. Since $q=p + \log(n)/\sqrt{n}$, a typical outcome of $\bA$ drawn from 
  $\Dno$ is slightly larger than a typical outcome drawn from $\Dyes$,
  and this difference causes almost every outcome of $|\bM \cup \bA|$ in $\Dno$ 
  (with $\bM\cup \bA$ being the set of relevant variables for $\boldf \leftarrow \Dno$) to be  larger than $k$
  by at least ${9}\sqrt{n}$.  
As a result, the relevant variables of any $k$-junta must miss either (a) at least one variable from $\bM$, or (b) at least ${9}\sqrt{n}$ variables from $\bA$.  Missing even a single variable from $\bM$ causes the $k$-junta to be far from $\boldf$ (this is made precise in Claim~\ref{ofofof} below).  On the other hand, missing ${9}\sqrt{n}$ variables from $\bA$ means that with probability at least $\Omega(\eps)$, at least one variable is missing from a typical~$\bS_i$ (recall that these are random $(\eps/\sqrt{n})$-dense subsets of $\bA$).  Because $\bh_i$ is a random function over~the variables in $\bS_i$, missing even a single variable would lead to a constant fraction of error when $\bh_i$ is the function determining the output of $\boldf$.

\begin{lemma}\label{main2222}
A function $\boldf\leftarrow\Dno$ is $\eps$-far from being a $k$-junta with probability $1 - o(1)$.
\end{lemma}

\begin{proof}
Fix any subset $M\subset [n]$ of size $t$, and we consider $\boldf = \boldf_{M, \bA, \bH}$ where $\bA $ and $\bH$ are sampled according to the procedure for $\Dno$. With probability $1 - o(1)$ over the choice of $\bA$,  we have
\begin{equation}\label{hehebb}
|\bA| \geq
  q\hspace{0.02cm}m-\frac{\delta \sqrt{n}\log n}{2}
  \ge \delta n+ {2\delta\sqrt{n}\log n}\quad\text{and}\quad
|\bM\cup \bA|\ge k+\delta\sqrt{n}\log n.
\end{equation}
We assume this is the case for the rest of the proof and fix any such set $A \subset \oM$.
It suffices to show that $\boldf=\boldf_{M,A,\bH}$ is $\eps$-far from $k$-juntas with probability
  $1-o(1)$, where $\bH$ is sampled according to the rest (steps 3 and 4) of the procedure for $\Dno$
  (by  sampling $\bS_i$ from $A$ and then $\bh_i$ over $\bS_i$).

The plan for the rest of the proof is the following. For each $V \subset M \cup A$ of size ${9\sqrt{n}}$,
we use 
$\bE_V$ to denote the size of the \emph{maximum} set of vertex-disjoint, $\boldf$-bichromatic edges along
directions in $V$ only.
We will prove the following claim:

\begin{claim}
\label{cl:V-avoids-M}
For each $V\subset M\cup A$ of size $9\sqrt{n}$,  
  we have $\bE_V \geq \eps\hspace{0.02cm}2^{n}$ with probability $1 - \exp(-2^{\Omega(n)})$.
\end{claim}

Note that when $\bE_V\ge \eps\hspace{0.02cm}2^{n}$, we have $\dist(\boldf, g)\ge \eps$ for every function $g$
  that does not depend on any variable in $V$.
This~is~because, for~every $\boldf$-bichromatic edge $(x,x^{(\ell)})$  along a coordinate $\ell\in V$,
  we must have $\boldf(x) \neq \boldf(x^{(\ell)})$ since the edge is bichromatic but $g(x) = g(x^{(\ell)})$ as $g$ does not depend on the $\ell$th variable.
As a result, $\boldf$ must disagree  with $g$ on at least $\eps\hspace{0.02cm} 2^n$ many points.

Assuming Claim \ref{cl:V-avoids-M} for now, we can apply a union bound over all 
$$
{|M \cup A| \choose 9\sqrt{n}} \leq 
\binom{n}{9\sqrt{n}} \leq 2^{O(\sqrt{n} \log n)}$$ possible choices of $V\subset M\cup A$ to conclude that with probability
  $1-o(1)$,
 $\boldf = \boldf_{M,A, \bH}$ is $\eps$-far from all functions that do not depend on
  at least $9\sqrt{n}$ variables in $M\cup A$.
By (\ref{hehebb}), this set includes all $k$-juntas.
This concludes the proof of the Lemma \ref{main2222} modulo the proof of Claim~\ref{cl:V-avoids-M}.
\end{proof}

In the rest of the section, we prove Claim \ref{cl:V-avoids-M}
  for a fixed subset $V \subset M \cup A$ of size $9\sqrt{n}$.
We start with the simpler case when $V\cap M$ is nonempty.

\begin{claim}\label{ofofof}
If $V\cap M\ne \emptyset$, then we have $\bE_V\ge 2^n/{5}$ with probability $1 - \exp(-2^{\Omega(n)})$.
\end{claim}
\begin{proof}
Fix an $\ell \in V \cap M$; we will argue that with probability $1 - \exp(-2^{\Omega(n)})$ there are at least~$2^n/5$
$\boldf$-bichromatic edges along direction $\ell$.  This suffices since such  edges are clearly vertex-disjoint.

Observe that since $\ell \in M$, every $x \in \{0, 1\}^n$ has $\Gamma(x) \neq \Gamma(x^{(\ell)})$. 
For each $b \in \{0, 1\}^{M}$,  let
$X_b$~be the set of $x \in \{0, 1\}^n$ with $x_{|S} = b$.
We partition $\{0,1\}^n$ into $2^{t-1}$ pairs 
  $X_b$ and $X_{b^{(\ell)}}$, where $b$ ranges over the $2^{t-1}$ strings in $\zo^M$ with $b_\ell=0.$
For each such pair, we use $\bD_b$
  to denote the number of $\boldf$-bichromatic edges between $X_b$ and $X_{b^{(\ell)}}$.
We are interested in lower bounding 
$
\sum_{b} \bD_b.
$

We will apply Hoeffding's inequality.  
For this purpose we note that  the $\bD_b$'s are independent~(since they depend on distinct
  $\bh_i$'s), always lie between $0$ and $2^{m}$, and each one has expectation $2^{m-1}$.
The latter is because each edge $(x,x^{(\ell)})$ 
  has $\boldf(x)$ and $\boldf(x^{(\ell)})$ drawn as two independent 
  random bits, which is the case since $\Gamma(x)\ne \Gamma(x^{(\ell)})$.
Thus, the expectation of $\sum_b \bD_b$ is $2^{n-2}$.
By Hoeffding's inequality, we have
$$
\Pr\left[ \hspace{0.03cm}\left|\sum \bD_b-2^{n-2}\right|\ge \frac{2^n}{20}
\hspace{0.03cm}\right]\le 2\cdot \exp\left(-\frac{2(2^n/20)^2}
{2^{t-1}\cdot 2^{2m}}\right)=\exp\left(-2^{\Omega(n)}\right)
$$  
since $t=\Omega(n)$.
This finishes the proof of the claim.
\end{proof}

Now we may assume that $V\subset A$ (and $|V|=9\sqrt{n}$). 
We use $\bI$ to denote the set of $i\in [N]$ such that 
  $\bS_i\cap V\ne \emptyset$. 
The following claim shows that $\bI$ is large with extremely high probability:
\begin{claim}
\label{cl:claim1}
We have $|\bI| \ge {4.4\hspace{0.02cm}\eps\hspace{0.01cm}N}$ with probability   
  at least $1 - \exp(-2^{\Omega(n)})$ over the choice of $\bS$.
\end{claim}
\begin{proof}
For each $i\in [N]$ we have (using $1-x\le e^{-x}$ for all $x$ 
  and $1-x/2\ge  e^{-x}$ for $x\in [0,1.5]$): 
\[ 
\Prx\big[ i\in \bI
\hspace{0.02cm}
\big] 
=1-\left(1-\frac{\eps}{\sqrt{n}}\right)^{9\sqrt{n}}
\ge 1-e^{-9\eps}\ge 4.5\hspace{0.02cm}\eps,
\]
since ${\eps}/{\sqrt{n}}$ is the probability of each element of $A$
  being included in $\bS_i$ and $\eps\le 1/6$ so $9\eps\le 1.5$.

Using $\eps\ge 2^{-(2\alpha-1)n/2}$ from (\ref{epsbound2}), we have $\E[\hspace{0.02cm}|\bI|\hspace{0.02cm}]\ge 
4.5\hspace{0.02cm}\eps \hspace{0.01cm}N =2^{\Omega(n)}$.
Since the $\bS_i$'s are~independent, a Chernoff bound implies that  $|\bI|\ge 4.4\hspace{0.02cm}\eps \hspace{0.01cm}N$ with probability $1-\exp(-2^{\Omega(n)})$.
\end{proof}

By Claim \ref{cl:claim1}, we fix $S_1,\ldots,S_N$ to be any sequence of 
  subsets of $A$ that satisfy $|I|\ge 4.4\hspace{0.02cm}\eps\hspace{0.01cm}N$ in~the rest of the proof, and it suffices to show that over the random choices 
  of $\bh_1,\ldots,\bh_N$ (where each $\bh_i$ is chosen to be a random function over $S_i$),~$\bE_V\ge \eps\hspace{0.02cm}2^n$ with probability at least $1-\exp(-2^{\Omega(n)})$.

To this end we use $\rho(i)$ for each $i\in I$ to denote the 
  first coordinate of $S_i$ in $V$, and 
$Z_{i}$ to denote the set of $x \in \{0, 1\}^n$ with $\Gamma(x) = i$.
Note that the $Z_i$'s are disjoint.
We further~partition each~$Z_i$ into disjoint $Z_{i,b}$, $b\in \{0,1\}^{S_i}$,
  with $x\in Z_{i,b}$ iff $x \in Z_i$ and $x_{|S_i}=b$.
For each $i\in I$ and $b\in \{0,1\}^{S_i}$ with $b_{\rho(i)}=0$,
  we use $\bD_{i,b}$ to denote the number of $\boldf$-bichromatic edges  
  between $Z_{i,b}$ and $Z_{i,b^{(\rho(i))}}$ along the $\rho(i)$th direction.
It is clear that such edges, over all $i$ and $b$, are vertex-disjoint and thus,
\begin{equation}\label{pfpfpf}
\bE_V\ge \sum_{i\in I} \hspace{0.05cm}\sum_{\substack{b\in \{0,1\}^{S_i}\\ b_{\rho(i)}=0}} \bD_{i,b}.
\end{equation}
  
We will apply Hoeffding's inequality. 
Note that $\bD_{i,b}$ is $2^{m-|S_i|}$ with probability $1/2$, and $0$ with probability $1/2$.
Thus, the expectation of the RHS of (\ref{pfpfpf}) is 
$$
\sum_{i\in I} 2^{|S_i|-1}\cdot 2^{m-|S_i|-1}=|I|\cdot 2^{m-2}\ge 1.1\hspace{0.02cm}\eps\hspace{0.02cm}2^n,
$$ 
using $|I|\ge 4.4\hspace{0.02cm}\eps\hspace{0.01cm}N$.
Since all the $\bD_{i,b}$'s are independent, by Hoeffding's inequality we have
\begin{align*}
\Pr\Big[ \hspace{0.03cm}\left|\hspace{0.03cm}\text{RHS of (\ref{pfpfpf})}-|I|\cdot 2^{m-2} \right|\ge 
  0.01\hspace{0.02cm}|I|\cdot 2^{m-2}
\hspace{0.03cm}\Big]&\le2\cdot \exp\left(-\frac{2(0.01\hspace{0.02cm}|I|\cdot 2^{m-2})^2}{\sum_{i\in I} 2^{|S_i|-1}
\cdot 2^{2(m-|S_i|)}}\right)
\\[0.5ex]&\le \exp\left(-2^{\Omega(n)}\right),
\end{align*} 
since $|I|\ge \Omega(\eps\hspace{0.01cm}N)=2^{\Omega(n)}$.
When this does not happen, we have $\bE_V\ge 0.99\cdot |I|\cdot 2^{m-2}>\eps\hspace{0.01cm}2^n$.  

This concludes the proof of Claim \ref{cl:V-avoids-M}. \qed

\section{The Set-Size-Set-Queries (SSSQ) Problem}
\label{sec:simple-task}


We first introduce the Set-Size-Set-Queries (\SSSQ\ for short) problem,
  which is an artificial problem that we use as a bridge to prove Theorem \ref{thm:main}.
We use the same parameters $p,q$ and $m$ from the definition of $\Dyes$ and $\Dno$, with
  $n$ being sufficiently large (so $m=\Omega(n)$ is sufficiently large as well).


We start by defining  $\Ayes$ and $\Ano$, two distributions over subsets of $[m]$:
  $\bA\sim \Ayes$
  is drawn by independently including each element of $[m]$ with probability $p$ and
  $\bA\sim\Ano$ is drawn by independently including each element with probability $\pn$.  
In \SSSQ, the algorithm needs to determine whether an unknown $A\subseteq [m]$ is drawn from $\Ayes$ or $\Ano$.
(For intuition, to see that this task is reasonable, we observe here that a straightforward Chernoff bound shows that almost every outcome of $\bA \sim \Ayes$ is larger than almost every outcome of $\bA \sim \Ano$
  by $\Omega(\sqrt{n}\log n)$.)

Let $A$ be a subset of $[m]$ which is hidden in an oracle.
An algorithm accesses $A$ (in order to tell whether it is drawn from $\Ayes$ or $\Ano$)
  by  interacting with the oracle in the following way:  each time it calls the oracle, it does so by sending a subset of $[m]$ to the oracle.  The oracle responds as follows:    for each $j$ in the subset, it returns a bit that is $0$ if $j\notin A$,
  and is $1$ with probability $\eps/\sqrt{n}$ and $0$ with probability $1-\eps/\sqrt{n}$ if $j\in A$.  The cost of such an oracle call is the size of the subset provided to the oracle.  
  
More formally, a deterministic and non-adaptive algorithm $\Alg=(g,T)$ for \SSSQ ~accesses the set
  $A$ hidden in the oracle by submitting  
  a list of queries $T=(T_1,\ldots, T_d)$, for some $d\ge 1$, where each $T_i \subseteq [m]$ is a set. 
  (Thus, we call each $T_i$ a \emph{set query}, as part of the name \SSSQ.)
\begin{flushleft}\begin{itemize}
\item  Given $T$, the oracle
  returns a list of random vectors $\bv=(\bv_1,\ldots,\bv_d)$, where $\bv_i \in \{0, 1\}^{T_i}$ and
  each bit $\bv_{i,j}$ is independently distributed as follows:  if $j\notin A$ then $\bv_{i,j}=0$,  and if $j \in A$ then
\begin{equation} \label{eq:itmfa}
\bv_{i,j} = \begin{cases}1 & \text{with probability $\eps /{\sqrt{n}}$} \\
						 0 & \text{with probability $1-(\eps /\sqrt{n})$}. \end{cases} 
\end{equation}
						 Note that the random vectors in $\bv$ depend on both $T$ and $A$.
\item Given $\bv=(\bv_1,\ldots,\bv_d)$, $\Alg$ returns (deterministically) the value of
  $g(\bv)\in \{\text{``yes''}, \text{``no''}\}$.
\end{itemize}\end{flushleft}
The performance of $\Alg=(g,T)$ is measured by its
 \emph{query complexity} and its \emph{advantage}.
\begin{flushleft}\begin{itemize}
\item The query complexity of $\Alg$ is defined as $\sum_{i=1}^d |T_i|$, the total size of all the set queries.
On the other hand, the {advantage} of $\Alg$ is defined as
\[
\Prx_{\bA \sim \Ayes}\big[\Alg(\bA)= \text{``yes''}\big]
-
\Prx_{\bA \sim \Ano}\big[\Alg(\bA)=\text{``yes''}\big]
.
\]

\end{itemize}\end{flushleft}


\begin{remark}\label{ob1}
In the definition above, $g$ is a deterministic map from all possible sequences of vectors returned by
  the oracle to ``yes'' or ``no.''
Considering only deterministic as opposed to randomized $g$ is without loss of generality since 
given any query sequence $T$, the highest possible advantage can always be
  achieved by a deterministic map $g$.
\end{remark}

We prove the following lower bound for any 
  deterministic, non-adaptive $\Alg$ in  Section~\ref{sec:proof-sssq}.

\begin{lemma}\label{lem:model1lb}
Any deterministic, non-adaptive $\Alg$ for $\emph{\SSSQ}$ with  
  advantage at least $2/3$ satisfies
$$\sum_{i=1}^d |T_i| \geq \frac{n^{3/2}}{\eps\cdot  \log^3 n \cdot \log^2(n/\eps)}.$$
\end{lemma}


\section{Reducing from SSSQ to distinguishing $\Dyes$ and $\Dno$}
\label{sec:reduction}

In this section we reduce from \SSSQ~to the problem of distinguishing the pair of distributions $\Dyes$ and $\Dno$.
More precisely, let $\Alg^*=(h,X)$ denote a deterministic and nonadaptive algorithm that makes
  $q\le (n/\eps)^2$ string queries\footnote{Any algorithm that makes more than this many queries 
  already fits the $\widetilde{\Omega}(n^{3/2}/\eps)$ lower bound we aim for.} $X=(x_1,\ldots,x_q)$ to a hidden function $f$ drawn from either $\Dyes$ or $\Dno$, applies the (deterministic) map $h$ to return
  $h(f(x_1),\ldots,f(x_q))\in \{\text{``yes''},\text{``no''}\}$, and satisfies
\begin{equation}\label{ueue}
\Prx_{\boldf \sim \Dyes}\big[ \Alg^*(\boldf) = \text{``yes''} \big] - \Prx_{\boldf \sim \Dno}\big[ \Alg^*(\boldf) = \text{``yes''}\big] \geq 3/4.
\end{equation}
We show how to define from $\Alg^*=(h,X)$ an algorithm $\Alg=(g,T)$ for the problem \SSSQ~with
  query complexity at most $\tau \cdot q$ and advantage $2/3$, where 
  $
\tau=c_\alpha \cdot 5 \log(n/\eps)
$
 and $$c_\alpha=-\frac{1}{\log (1.5-\alpha)} >0
 \quad\text{with}\quad \text{$(1.5-\alpha)^{c_\alpha}=1/2$}$$ is a constant that depends on $\alpha$. 
Given this reduction it follows from Lemma \ref{lem:model1lb} that
$q\ge \widetilde{\Omega}(n^{3/2}/\eps).$
This finishes the proof of Theorem \ref{thm:main}.

\def\Eyes{\mathcal{E}_\text{yes}} \def\Eno{\mathcal{E}_\text{no}}

We start with some notation.
Recall that in both $\Dyes$ and $\Dno$, 
  $\bM$ is a subset of $[n]$ of size $t$ drawn uniformly at random.
For a fixed $M$ of size $t$, we use $\Eyes(M)$ to denote the distribution of
  $\bA$ and $\bH$ sampled in the randomized procedure for $\Dyes$, conditioning on $\bM=M$.
We define $\Eno(M)$ similarly.
Then conditioning on $\bM=M$, $\boldf\sim \Dyes$ 
  is distributed as $f_{M,\bA,\bH}$ with $(\bA,\bH)\sim \Eyes(M)$ and
  $\boldf\sim \Dno$ is distributed as $f_{M,\bA,\bH}$ with $(\bA,\bH)\sim \Eno(M)$. 
This allows us to rewrite (\ref{ueue}) as
\begin{equation*}
\frac{1}{{n\choose t}}\cdot \sum_{M : |M|=t} \left(
  \Prx_{(\bA,\bH) \sim \Eyes(M)}\big[ \Alg^*(\boldf_{M,\bA,\bH}) = \text{``yes''} \big] - \Prx_{(\bA,\bH) \sim \Eno(M)}\big[ \Alg^*(\boldf_{M,\bA,\bH}) = \text{``yes''}\big]\right) \geq \frac{3}{4}.
\end{equation*}
 %

We say $M\subset [n]$ is \emph{good} if 
  any two queries $x_i$ and $x_j$ in $X$ with Hamming distance $\|x_i-x_j\|_1 \geq \tau$  have different projections on $M$, i.e., $(x_i)_{|M} \neq (x_j)_{|M}$.
We prove below that most $M$'s are good.

\begin{claim}\label{easyeasy}
$\Pr_{\bM}\big[\hspace{0.02cm}\bM\text{~is not good}\hspace{0.05cm}\big]=o(1).$ 
\end{claim}
\begin{proof}
For each pair of strings $x_i$ and $x_j$ in $X$ with Hamming distance at least
  $\tau\ignore{ = \red{8\log (n/\eps)}}$, the probability of them having
  the same projection on $\bM$ (drawn uniformly from all size-$t$ subsets) is at most
\begin{align*}
\frac{{n-\tau\choose t}}{{n\choose t}}&=\frac{(n-\tau-t+1)\cdots (n-t)}{(n-\tau+1)
\cdots n}
\leq \left(1 - {\frac t n}\right)^{\tau}
\le \big(2(1-\alpha) + o(1)\big)^\tau <  \left(1.5-\alpha\right)^{\tau}
 \leq O\left(\frac{\eps}{n}\right)^5, 
\end{align*}
by our choices of $c_\alpha$ and $\tau$.
The claim follows by a union bound over at most $q^2\le (n/\eps)^4$ pairs.
\end{proof}

We can split the sum (\ref{ueue}) into two sums:
  the sum over good $M$ and the sum over bad $M$.  By Claim \ref{easyeasy}
    the contribution from the bad $M$ is at most $o(1)$, and thus we have that
\begin{equation*}
\frac{1}{{n\choose t}}\cdot \sum_{\text{good}\ M} \left(
  \Prx_{(\bA,\bH) \sim \Eyes(M)}\big[ \Alg^*(\boldf_{M,\bA,\bH}) = \text{``yes''} \big] - \Prx_{(\bA,\bH) \sim \Eno(M)}\big[ \Alg^*(\boldf_{M,\bA,\bH}) = \text{``yes''}\big]\right) 
\end{equation*}
is at least $3/4 - o(1)$.
Thus, there must exist a good set $M \subset [n]$ of size $t$ with
\begin{align}
\label{eq:m-and-alg}
\Prx_{(\bA, \bH) \sim \Eyes(M)}\big[ \Alg^*(\boldf_{M, \bA, \bH}) = \text{``yes''}\big] - \Prx_{(\bA, \bH) \sim \Eno(M)}\left[ \Alg^*(\boldf_{M, \bA, \bH}) = \text{``yes''}\right] \geq 2/3. 
\end{align}
Fix such a good $M$. We use $\Alg^*=(h,X)$ and $M$ to define an algorithm $\Alg=(g,T)$ for 
  \SSSQ~as follows (note that 
  the algorithm $\Alg$ below actually works over the universe $\overline{M}$ (of size $m$) 
  instead of $[m]$ as in the original definition of \SSSQ\ but this
  can be handled by picking any bijection between $\overline{M}$ and $[m]$;
  accordingly $\bA\sim \Ayes$~is drawn by including each element of $\overline{M}$ with probability $p$
  and $\bA\sim\Ano$ is drawn by including each element of $\overline{M}$ with probability $q$).
We start with $T$:
\begin{flushleft}\begin{enumerate}
\item First we use $M$ to define an equivalence relation $\sim$ over the query set $X$, where
  $x_i\sim  x_j$ if $(x_i)_{|M}=(x_j)_{|M}$. 
Let $X_1,\ldots,X_d$, $d\ge 1$, denote the equivalence classes of $X$,
  and let us write $\rho(\ell)$ for each $\ell\in [d]$ to denote the value $\Gamma(x)\in [N]$ that is shared by all strings $x\in X_\ell$. 
\item Next we define a sequence of subsets of $\overline{M}$, $T=(T_1,\ldots,T_d)$, as the set queries of
  $\Alg$, where
\begin{align}
\label{eq:def_t}
T_\ell &= \big\{ i \in \overline{M} \colon \exists\hspace{0.05cm} x,y \in X_\ell \text{~such that~} x_i \neq y_i \big\}. 
\end{align}
\end{enumerate}\end{flushleft}
To upper bound $|T_\ell|$, fixing an arbitrary string $x\in X_\ell$ and recalling that $M$ is good, we have that
$$
|T_\ell|\le \sum_{y\in X_\ell} \|x-y\|_1 \le \sum_{y \in X_\ell} \tau
=\tau\cdot |X_\ell| .
$$
As a result, the query complexity of $\Alg$ (using $T$ as its set queries) is at most
$$
\sum_{\ell=1}^d |T_\ell| \le \tau\cdot \sum_{\ell=1}^d|X_\ell|\le \tau\cdot q.
$$

It remains to define $h$ and then prove that the advantage of $\Alg=(g,T)$ for \SSSQ~is at least $2/3$.
Indeed the $g$ that we define is a randomized map and we describe it as a randomized procedure below
  (by Remark \ref{ob1} one can extract from $g$ a deterministic map that achieves 
  the same advantage):
\begin{flushleft}\begin{enumerate}
\item Given $v_1,\ldots,v_d$, $v_\ell\in \{0,1\}^{T_\ell}$, as the 
  strings returned by the oracle upon being given $T$,
let
\begin{equation}\label{fjfjf}
R_\ell=\big\{j\in T_\ell: v_{\ell,j}=1\big\}.
\end{equation}
For each $\ell \in [d]$, the procedure draws a random function $\boldf_\ell:\{0,1\}^{R_\ell}\rightarrow\{0,1\}$,
by flipping $2^{|R_\ell|}$ many independent and unbiased random bits.

\item Next for each query $x\in X_\ell$, $\ell\in [d]$, 
  we feed $\boldf_\ell(x_{|R_\ell})$ to $h$ as the bit that the oracle returns upon
  the query $x$.
Finally the procedure returns the result (``yes'' or ``no'') that $h$ returns.
\end{enumerate}\end{flushleft}
In the rest of the proof we show that the advantage of $\Alg=(g,T)$ is exactly the same as the
  LHS of (\ref{eq:m-and-alg}) and thus, is at least $2/3$.
  
  For convenience, we use $\Vyes$ to denote the distribution of responses $\bv=(\bv_1,\ldots,\bv_d)$ to $T$ when $\bA \sim\Ayes$, and $\Vno$ to denote the distribution when $\bA \sim\Ano$.
Then the advantage of $\Alg$ is
\[
\Prx_{\bv \sim \Vyes}\big[\hspace{0.03cm}g(\bv)=\text{``yes''}\big]
-
\Prx_{\bv \sim \Vno}\big[\hspace{0.03cm}g(\bv)=\text{``yes''}\big]
.
\]
It suffices to show that \vspace{0.08cm}
\begin{align}\label{yescase}
\Prx_{\bv \sim \Vyes}\big[\hspace{0.03cm}g(\bv)=\text{``yes''}\big]&=\Prx_{(\bA, \bH) \sim \Eyes(M)}\big[ \Alg^*(\boldf_{M, \bA, \bH}) = \text{``yes''}\big] \quad\text{\ and} \\[0.2ex]
\Prx_{\bv \sim \Vno}\big[\hspace{0.03cm}g(\bv)=\text{``yes''}\big]&=
\Prx_{(\bA, \bH) \sim \Eno(M)}\left[ \Alg^*(\boldf_{M, \bA, \bH}) = \text{``yes''}\right].
\label{nocase}
\end{align}
We show (\ref{yescase}); the proof of (\ref{nocase}) is similar.
From the definition of $\Vyes$ and $\Eyes(M)$ the distribution~of $(\mathbf{R}_\ell:\ell\in [d])$
  derived from $\bv\sim \Vyes$ using (\ref{fjfjf}) is the same as the distribution of 
  $(\bS_{\rho(\ell)}\cap T_\ell:\ell\in [d])$:  both are sampled by first drawing a random subset $\bA$
  of $\smash{\overline{M}}$ and then drawing a random subset of $\bA\cap T_\ell$ independently by including each element of $\bA\cap T_\ell$ with the same probability ${\eps}\ignore{\eps_0}/\sqrt{n}$ (recall in particular equation (\ref{eq:itmfa}) and
  step~3 of the randomized procedure specifying $\Dyes$ in Section~\ref{sec:distributions}).
Since $\boldf_{M,\bA,\bH}(x)$ for $x\in X_\ell$ is determined by a random Boolean function $\bh_{\rho(\ell)}$
  from $\{0,1\}^{\bS_{\rho(\ell)}}$ to $\{0,1\}$,
  and since all the queries in $X_\ell$ only differ by coordinates in $T_\ell$,
  the distribution of the $q$ bits that $g$ feeds to $h$ when $\bv\sim \Vyes$ is the same
  as the distribution of $(\boldf(x):x\in X)$ when $\boldf\sim \Eyes(M)$.
This finishes the proof of (\ref{yescase}), and concludes our reduction argument.

\def\Ayes{\mathcal{A}_{\text{yes}}}
\def\Ano{\mathcal{A}_{\text{no}}}
\def\py{p}
\def\pn{q}
\def\Alg{\textsc{Alg}}
\def\Vyes{\mathcal{V}_{\text{yes}}}
\def\Vno{\mathcal{V}_{\text{no}}}
\def\Byes{\mathcal{B}_{\text{yes}}}
\def\Bno{\mathcal{B}_{\text{no}}}

\section{A lower bound on the non-adaptive query complexity of SSSQ}
\label{sec:proof-sssq}


\ignore{
Let $\alpha\in (1/2,1)$ be a constant and $\eps(n)$ be a function that satisfies (\ref{epsbound}) as
  in Theorem \ref{thm:main}.
In the~rest of the paper we choose $n$ to be a sufficiently large integer and
  set $k=\alpha n$ and $\eps=\eps(n)$. Let $\kappa = 120$ and $\eps_0=\kappa\eps$. We also let $\delta=1-\alpha\in (0,1/2)$ be a positive
  constant and
$$\color{red} m=2\delta n+{\delta\sqrt{n}\log n} .$$
In this section we describe two algorithmic tasks, \SSSQ ~and \SSEQ, and prove lower bounds on the query complexity of these tasks. We will then show
  in Section \ref{sec:reduction} that a ``too-good-to-be-true'' non-adaptive deterministic algorithm for distinguishing
  $\Dyes$ and $\Dno$ (to be defined in Section \ref{sec:distributions})
  can be converted into a query-efficient deterministic algorithm for \SSSQ, thus contradicting the lower bounds.

For the rest of the paper, fix
\[ p \eqdef \frac{1}{2}\quad\text{and}\quad q \eqdef \frac{1}{2} + \frac{\log n}{\sqrt{n}}. \]

Let $\Ayes$ be the distribution over subsets of $[m]$ obtained by independently including each element with probability $p$, and let $\Ano$ be the distribution over subsets of $[m]$ obtained by independently including each element with probability $\pn$.  Both tasks we define below are to determine whether an unknown subset $A =\bA \subseteq [m]$ is drawn from $\Ayes$ or $\Ano.$
\red{(For intuition to see that the tasks are reasonable, we observe here that a straightforward Chernoff bound shows that almost every outcome of $\bA \sim \Ayes$ is larger than almost every outcome of $\bA \sim \Ano$.)}
The differences between the two tasks \SSSQ~and \SSEQ~are 1) how an algorithm accesses $A$,  and 2) how we measure the query complexity of an algorithm.

\subsection{Set-Size-Set-Queries (\SSSQ)}\label{task1def}

Let $A$ be a subset of $[m]$ (drawn from $\Ayes$ or $\Ano$).
A (deterministic and non-adaptive) algorithm $\Alg=(g,T)$ for \SSSQ ~accesses $A$ by submitting to an oracle 
  a list of queries $T=(T_1,\ldots, T_k)$ for some $k\ge 1$, where each $T_i \subseteq [m]$ is a set. Thus, we call each $T_i$ a \emph{set query}.
\begin{itemize}
\item  On receiving $T$, the oracle
  returns a sequence of random vectors $\bv=(\bv_1,\ldots,\bv_k)$, where $\bv_i \in \{0, 1\}^{T_i}$ and
  for each $i \in [k]$ and $j\in T_i$, $\bv_{i,j}$ is independently distributed as follows:  if $j\notin A$ then $\bv_{i,j}=0$,  and if $j \in A$ then
\begin{equation} \label{eq:itmfa}
\bv_{i,j} = \begin{cases}1 & \text{with probability ${\color{red}\eps_0}/{\sqrt{n}}$} \\
						 0 & \text{with probability $1-({\color{red}\eps_0}/\sqrt{n})$}. \end{cases} 
\end{equation}
						 Note that $\bv$ is a sequence of random vectors that depends on both $T$ and $A$.
\item After receiving $\bv=(\bv_1,\ldots,\bv_k)$, $\Alg$ returns (deterministically) the value
  $g(\bv)\in \{\text{``yes''}, \text{``no''}\}$.
\end{itemize}

The performance of $\Alg=(g,T)$ is measured by its
 \emph{query complexity} and its \emph{advantage}.\rnote{Renamed this ``advantage'' instead of ``success probability'' -- since it's a difference of two probabilities and could even be negative, it felt weird calling it a ``probability'' :)}
\begin{itemize}

\item The query complexity of $\Alg$ is defined as $\sum_{i=1}^k |T_i|$, the total size of all the set queries.

\item The {advantage} of $\Alg$ is defined as

\[
\Prx_{\bA \sim \Ayes}\big[\Alg(\bA) \text{~outputs ``yes''}\big]
-
\Prx_{\bA \sim \Ano}\big[\Alg(\bA) \text{~outputs ``yes''}\big]
.
\]
For simplicity, we use $\Vyes$ to denote the distribution of responses $\bv$ to $T$ when $\bA \sim\Ayes$, and $\Vno$ to denote the distribution when $\bA \sim \Ano$.
Then the advantage can be written as
\[
\Prx_{\bv \sim \Vyes}\big[\hspace{0.03cm}g(\bv)=\text{``yes''}\big]
-
\Prx_{\bv \sim \Vno}\big[\hspace{0.03cm}g(\bv)=\text{``yes''}\big]
.
\]
\end{itemize}


\begin{remark}\label{ob1}
In the definition above, $g$ is a deterministic map from all possible sequences of vectors returned by
  the oracle to ``yes'' or ``no.''
Considering only deterministic as opposed to randomized $g$ is without loss of generality since 
given a query sequence $T$, the highest possible advantage can always be
  achieved by a deterministic map $g$.
The same holds for algorithms for the second task, \SSEQ.
\end{remark}
}

We will prove Lemma~\ref{lem:model1lb} by first giving a reduction from an even simpler
  algorithmic task, which we describe next in Section \ref{task2def}.
We will then prove a lower bound for the simpler task in Section \ref{task2proof}.

\subsection{Set-Size-Element-Queries (\SSEQ)}\label{task2def}

Recall the parameters $m,p,q$ and $\eps$ and the two distributions $\Ayes$ and $\Ano$ 
  used in the definition of problem \SSSQ.
We now introduce a simpler algorithmic task called the Set-Size-Element-Queries (\SSEQ) problem using
  the same parameters and distributions.
  
Let $A$ be a subset of $[m]$ hidden in an oracle.
An algorithm accesses the oracle to tell whether~it is~drawn from $\Ayes$ or $\Ano$.
The difference between \SSSQ\ and \SSEQ\ is the way an algorithm accesses $A$.
In \SSEQ, an algorithm $\Alg'=(h,\ell)$ submits a vector $\ell=(\ell_1,\ldots,\ell_m)$ of nonnegative integers.
\begin{flushleft}\begin{itemize}
\item 
On receiving $\ell$, the oracle returns a random response vector $\bb \in \{0, 1\}^{m}$, where each entry $\bb_i$ is distributed independently as follows:  if $i \notin A$ then $\bb_i=0$, and if $i \in A$ then
\begin{align*}
\bb_i &=   \begin{cases} 1 & \text{with probability $\lambda(\ell_i)$} \\
						0 & \text{with probability $1-\lambda(\ell_i)$} \end{cases},\quad\text{\ \ where\ }
			\lambda(\ell_i) = 1 - \left(1 - \frac{\eps }{\sqrt{n}}\right)^{\ell_i}.
\end{align*}
Equivalently, for each $i\in A$, the oracle independently flips $\ell_i$ coins, each of which is $1$
  with probability $\eps /\sqrt{n}$, and at the end returns
  $\bb_i=1$ to the algorithm if and only if at least one of the coins is $1$. Thus, we refer to each $\ell_i$ as 
 $\ell_i$ \emph{element-queries} for the $i$th element.
\item After receiving the vector $\bb$ from the oracle, $\Alg'$ returns the value $h(\bb)\in \{\text{``yes''},\text{``no''}\}$.  Here $h$ is a deterministic map from $\zo^m$ to $ \{\text{``yes''},\text{``no''}\}$.
\end{itemize}\end{flushleft}
Similar to before, the performance of $\Alg'$ is measured by its query complexity and its advantage:
\begin{flushleft}\begin{itemize}
\item The query complexity of $\Alg'=(h,\ell)$ is defined as $\|\ell\|_1 = \sum_{i=1}^{m} \ell_{i}$.
For its advantage, we let $\Byes$ denote the distribution of response vectors $\bb$ to query $\ell$ when $\bA \sim\Ayes$, and
  $\Bno$ denote the distribution when $\bA \sim \Dno$.
The advantage of $\Alg'=(h,\ell)$ is then defined as
\[
\Prx_{\bb \sim \Byes}\big[\hspace{0.03cm}h(\bb)=\text{``yes''}\big]
-
\Prx_{\bb \sim \Bno}\big[\hspace{0.03cm}h(\bb)=\text{``yes''}\big]
.
\]
\end{itemize}\end{flushleft}

\begin{remark}\label{monotone}
It is worth pointing out (we will use it later) that the highest possible advantage over all deterministic maps $h$ is a monotonically non-decreasing
function of the coordinates of $\ell$.
To see this, let $A$ be the underlying set and let
  $\ell$ and $\ell'$ be two vectors with $\ell_i\le \ell_i'$ for every $i\in [m]$.
Let $\bb$ and $\bb'$ be the random vectors returned by the oracle upon $\ell$ and $\ell'$.
Then we can define $\bb^*$ using $\bb'$ as follows:
  $\bb^*_i=0$ if $\bb'_i=0$; otherwise when $\bb'_i=1$, we set
$$
\bb^*_i=\begin{cases} 1 & \text{with probability $\lambda(\ell_i)/\lambda(\ell_i')$}\\[0.6ex] 0&
\text{with probability $1-\lambda(\ell_i)/\lambda(\ell_i')$} \end{cases}.
$$
One can verify that the distribution of $\bb$ is exactly the same as the distribution of $\bb^*$.  Hence there is a randomized map
$h'$ such that the advantage of $(h',\ell')$ is at least as large as the highest possible
  advantage achievable using $\ell.$  The remark now follows by our earlier observation
  in Remark \ref{ob1} that the highest possible advantage using $\ell'$ is always 
   achieved by a deterministic $h'$.
\end{remark}

The following lemma reduces the proof of Lemma \ref{lem:model1lb} to proving a lower bound  for \SSEQ.

\begin{lemma}
\label{lem:one-to-two}
Given any deterministic and non-adaptive algorithm $\Alg=(g,T)$ for \emph{\SSSQ},
  there is~a deterministic and non-adaptive algorithm $\Alg'=(h,\ell)$ for \emph{\SSEQ}\ with the same query complexity~as $\Alg$ and advantage at least as large as that of
  $\Alg$.
\end{lemma}

\begin{proof}
We show how to construct $\Alg'=(h ,\ell)$ from $\Alg=(g,T)$, where
  $h $ is a randomized map,
  such that $\Alg'$ has exactly the same query complexity and advantage as those of $\Alg$.
The lemma then follows from the observation we made earlier in Remark \ref{ob1}.

We define $\ell$ first.
Given $T=(T_1,\ldots,T_d)$ for some $d\ge 1$, $\ell=(\ell_1,\ldots,\ell_m)$ is defined as
\[ \ell_{j} = \big|\{ i \in [d] : j \in T_i \} \big|. \]
So $\|\ell\|_1=\sum_{i=1}^d |T_i|$.
To define $h$ we describe a randomized procedure $P$ that,
  given any $b\in \{0,1\}^m$, outputs a sequence of random vectors $\bv=(\bv_1,\ldots,\bv_d)$ such that the following claim holds.

\begin{claim}
\label{cl:same-dist}
If $\bb\sim \Byes$ \emph{(}or $\Bno$\emph{)}, then $P(\bb)$ is distributed the same as $\Vyes$ 
  \emph{(}or $\Vno$, respectively\emph{)}.
\end{claim}

Assuming Claim \ref{cl:same-dist}, we can set $h=g\circ P$ and the advantage of $\Alg'$ would be the same as
  that of $\Alg$.
In the rest of the proof,
  we describe the randomized procedure $P$ and prove Claim \ref{cl:same-dist}.

Given $b\in \{0,1\}^m$, $P$ outputs a sequence of random vectors $\bv=(\bv_1,\ldots,\bv_d)$ as follows:
\begin{flushleft}\begin{itemize}
\item If $b_j = 0$, then for each $i\in [d]$ with $j\in T_i$, $P$ sets $\bv_{i,j} = 0$.
\item If $b_j = 1$ (this implies that $\ell_j>0$ and $j\in T_i$ for some $i\in [d]$), $P$
sets $(\bv_{i,j}:i\in [d],j\in T_i)$ to be a length-$r$, where $r=|\{i \in [d]: j \in T_i\}|$, binary string in which each bit is independently 1 with probability $\epsilon /\sqrt{n}$ and 0 with probability $1-\epsilon /\sqrt{n}$, conditioned on its not being $0^r$. 
\end{itemize}\end{flushleft}
\begin{proof}[Proof of Claim \ref{cl:same-dist}]
It suffices to prove that, fixing any
  $A\subseteq [m]$ as the underlying set hidden in the oracle, the distribution of $\bv$ is the same
  as the distribution of $P(\bb)$.
The claim then follows since in the definitions of both $\smash{\Byes}$ and $\smash{\Vyes}$
  (or $\smash{\Bno}$ and $\smash{\Vno}$), $A$ is drawn from $\smash{\Ayes}$ (or $\smash{\Ano}$, respectively).

Consider~a sequence
  $v$ of $d$ vectors $v_1, \dots, v_d$ with $v_i \in \{0, 1\}^{T_i}$ for each $i\in [d]$, and let
  $$n_{j, 1} = |\{ i \in [d] :j\in T_i\ \text{and}\ v_{i,j} = 1 \}|\quad\text{and}\quad
  n_{j, 0} = |\{ i \in [d] : j\in T_i\ \text{and}\  v_{i,j} = 0 \}|,$$
for each $j\in [m]$. Then  the $\bv$ returned by the oracle (in \SSSQ) is equal to $v$ with probability:
\begin{align}\label{pfpff}
\ind\big\{ \forall j \notin A,\hspace{0.04cm} n_{j,1}=0 \big\} 
									\cdot	\prod_{j \in A} \left( \frac{\eps}{\sqrt{n}}\right)^{n_{j, 1}} \left(1 - \frac{\eps}{\sqrt{n}} \right)^{n_{j, 0}},
\end{align}
since all coordinates $\bv_{i,j}$ are independent. On the other hand, the probability of $P(\bb)=v$ is
\begin{align*}
 \ind\big\{ \forall j \notin A,\hspace{0.04cm} n_{j,1} = 0\big\} \cdot 
									  \prod_{j \in A} \left( \ind\big\{n_{j, 0} = \ell_j \big\} \cdot \left(1 - \frac{\eps}{\sqrt{n}} \right)^{\ell_j}  
									   +\hspace{0.06cm} \ind\big\{n_{j, 1} \geq 1\big\}\cdot \left( \frac{\eps}{\sqrt{n}}\right)^{n_{j, 1}} \left( 1 - \frac{\eps}{\sqrt{n}}\right)^{n_{j, 0}}\right),
\end{align*}
which is exactly the same as the probability of $\bv=v$ in (\ref{pfpff}).\end{proof}

This finishes the proof of Lemma \ref{lem:one-to-two}.
\end{proof}

\subsection{
A lower bound for SSEQ}\label{task2proof}

We prove the following lower bound for \SSEQ, from which Lemma~\ref{lem:model1lb} follows:

\begin{lemma}\label{lem:model12}
Any deterministic, non-adaptive $\Alg'$ for \emph{\SSEQ} with  
  advantage at least $2/3$ satisfies
$$\|\ell\|_1> s\eqdef\frac{n^{3/2}}{\eps\cdot  \log^3 n \cdot \log^2(n/\eps)}.$$
\end{lemma}

\begin{proof}
Assume for contradiction that there is an algorithm $\Alg'=(h,\ell)$
  with  $\|\ell\|_1\le s$ and advantage at least $2/3$.
Let $\ell^*$ be the vector obtained from $\ell$ by rounding each positive $\ell_i$ to
  the smallest power of $2$ that is at least as large as $\ell_i$
  (and taking $\ell_i^*=0$ if $\ell_i=0$).
From Remark \ref{monotone}, there must be a map $h^*$ such that $(h^*,\ell^*)$
  also has advantage at least $2/3$ but now we have 1) $\|\ell^*\|_1\le 2s$
  and 2) every positive entry of $\ell^*$ is a power of $2$.
Below we abuse notation and still use $\Alg'=(h,\ell)$~to denote $(h^*,\ell^*)$: $\Alg'=(h,\ell)$
  satisfies $\|\ell\|_1\le 2s$,
  every positive entry of $\ell$ is a power of $2$, and has
advantage at least $2/3$.
We obtain a contradiction below by showing that any such $\ell$ can only
  have an advantage of $o(1)$.

Let $L=\lceil \log (2s)\rceil =O(\log (n/\eps))$.
Given that $\|\ell\|_1\le 2s$ we can 
  partition $\{i\in [m]:\ell_i>0\}$ into $L+1$ sets $C_0,\ldots,C_L$, where bin $C_j$  contains those coordinates $i \in [m]$ with $\ell_i=2^{j}$.
We may make two further assumptions on $\Alg'=(h,\ell)$ that will simplify the lower bound proof:
\begin{flushleft}\begin{itemize}
\item We may reorder the entries in decreasing order and assume without loss of generality that
\begin{equation}\label{eqhehe}
\ell = \left( \underbrace{2^{L}, \dots, 2^{L}}_{c_L}, \underbrace{2^{L-1}, \dots, 2^{L-1}}_{c_{L-1}} , \dots, \underbrace{1, \dots, 1}_{c_0}, 0, \dots, 0\right),
\end{equation}
where $c_j=|C_j|$ satisfies $\sum_j c_j\cdot 2^j\le 2s$.  This is without loss of generality since $\Ayes$ and $\Ano$ are symmetric in the coordinates (and so are $\Byes$ and $\Bno$).
\item For the same reason we may assume that the map $h(b)$ depends only on the number of $1$'s of $b$ in
  each set $C_j$,  which we refer~to as the \emph{summary} $S(b)$ of $b$:
$$
S(b)\eqdef \Big(\|b_{|C_L}\|_{1}, \|b_{|C_{L-1}}\|_1, \dots, \|b_{|C_0}\|_1\Big)\in \mathbb{Z}_{\ge 0}^{L+1}.
$$
To see that this is without loss of generality, consider a randomized procedure $P$ that, given $b\in \{0,1\}^m$, 
applies an independent random permutation
  over the entries of $C_j$ for each bin $j\in [0:L]$.
One can verify that the random map $h'=h\circ P$ only depends on
  the summary $S(b)$ of $b$ but achieves the same advantage as $h$.
\end{itemize}\end{flushleft}

\def\Syes{\mathcal{S}_{\text{yes}}}
\def\Sno{\mathcal{S}_{\text{no}}}


Given a query $\ell$ as in (\ref{eqhehe}),
  we define $\Syes$ to be the distribution of $S(\bb)$ for $\bb\sim \Byes$
  (recall that $\Byes$ is the distribution of the vector $\bb$ returned by the oracle
  upon the query $\ell$ when $\bA\sim\Ayes$).
Similarly we define $\Sno$ as the distribution of $S(\bb)$ for $\bb\sim \Bno$.
As $h$ only depends on the summary the advantage is at most $\dtv(\Syes, \Sno)$,
  which we upper bound below by $o(1)$.

From the definition of $\Byes$ (or $\Bno$, respectively) and the fact that
  $\Ayes$ (or $\Ano$, respectively) is symmetric over the $m$ coordinates,
  we have that the $L+1$ entries of $\Syes$ (of $\Sno$, respectively) are mutually independent, and that 
  their entries for each $C_j$, $j\in [0:L]$, are distributed as
  $\Bin(c_j, \py \lambda_j)$ (as $\Bin(c_j, \pn \lambda_j)$, respectively), where we have
$
\smash{\lambda_j = 1 - ( 1 - ( {\eps}/{\sqrt{n}}))^{2^j}.}
$

In order to prove that $\dtv(\Syes, \Sno)=o(1)$ and achieve the desired contradiction, we will give upper bounds on the total variation distance between their $C_j$-entries for each $j\in \{0, \dots, L\}$.
\begin{claim}
\label{lem:dtvbound}
Let $\bX \sim \Bin(c_j, \py \lambda_j)$ and $\bY \sim \Bin(c_j, \pn \lambda_j)$.
Then 
$\dtv(\bX, \bY) \leq o\left(1/L\right).$
\end{claim}

We delay the proof of Claim \ref{lem:dtvbound}, but assuming
  it we may simply apply the following well-known proposition to conclude that
$\dtv(\Syes, \Sno) = o(1)$.
\begin{proposition} [Subadditivity of total variation distance]
Let $\bX = (\bX_1, \dots, \bX_k)$ and $\bY = (\bY_1,$ $\dots, \bY_k)$ be two tuples of independent random variables. Then $\dtv(\bX, \bY) \leq \sum_{i=1}^{k} \dtv(\bX_i, \bY_i)$.
\end{proposition}

This gives us a contradiction and finishes the proof of Lemma \ref{lem:model12}.
\end{proof}

\noindent Below we prove Claim \ref{lem:dtvbound}.

\begin{proof}[Proof of Claim~\ref{lem:dtvbound}]
The claim is trivial when $c_j=0$ so we assume below that $c_j>0$.

Let $r = \py\lambda_j$ and $x =  {\log n \cdot \lambda_j }/{\sqrt{n}}$. Then $\bX \sim \Bin(c_j, r)$ and $\bY \sim \Bin(c_j, r + x)$. As indicated in Equation (2.15) of \cite{adell06}, Equation (15) of \cite{Roos:00} gives
\begin{align}
\dtv(\bX, \bY) &\leq O\left( \frac{\tau(x)}{(1-\tau(x))^2} \right),\quad\text{where}\quad
\tau(x)\eqdef x \sqrt{\dfrac{c_j + 2}{2r(1 - r)}}, \label{eq:dtv-bound}
\end{align}
whenever $\tau(x)<1$. Substituting for $x$ and $r$, we have (using $c_j\ge 1$, $r\le 1/2$ and $p=1/2$)
\[ \tau(x) =
O\left(\frac{\log n\cdot \lambda_j}{\sqrt{n}}\cdot \sqrt{\frac{c_j}{r}}\right)=
O\left(\log n \cdot \sqrt{\dfrac{\lambda_j\cdot c_j}{n}} \right) = O\left( \dfrac{1}{L}
\cdot\sqrt{\dfrac{n^{1/2}\cdot  \lambda_j}{2^j\cdot  \eps\cdot \log n}} \right), \]
where the last inequality follows from
$$c_j\cdot  2^j \leq 2s \leq O\left(\dfrac{ n^{3/2}}{\eps\cdot \log^3 n \cdot L^2}\right).$$
Finally, note that (using $1-x>e^{-2x}$ for small positive $x$ and $1-x\le e^{-x}$ for all $x$):
$$
1 - \lambda_j = 
\left( 1 - \frac{\eps }{\sqrt{n}} \right)^{2^j}
\ge \left(e^{-2 \eps/\sqrt{n}}\right)^{2^j}=e^{-2^{j+1}  \eps/\sqrt{n}}\ge 1-O(2^j\eps/\sqrt{n})
$$
and  
$\dfrac{\sqrt{n}\cdot \lambda_j}{2^j \cdot\eps}=O(1)$.
This implies $\tau(x)=o(1/L)=o(1)$.
The claim then follows from (\ref{eq:dtv-bound}).
\end{proof}

\begin{flushleft}
\bibliographystyle{alpha}
\bibliography{allrefs}
\end{flushleft}
\appendix

\section{Proof of Theorem~\ref{thm:main-intro} assuming Theorem~\ref{thm:main}} \label{app:1}


We prove the following claim in Appendix \ref{app:simpleclaim}.
\begin{claim}\label{simpleclaim}
Let $\eps(n)$ be a  function that satisfies $2^{-n}\le \eps(n)\le 1/5$ for sufficiently large $n$.
Then~any non-adaptive algorithm that accepts the all-$0$ function with probability at least $5/6$ and rejects every
  function that is $\eps$-far from $(n-1)$-juntas with probability
  at least $5/6$ must make $\Omega(1/\eps)$ queries.
\end{claim}
Next let $k(n)$ and $\eps(n)$ be the pair of functions from the statement of Theorem \ref{thm:main-intro}.
We consider a sufficiently large $n$ (letting $k=k(n)$ and $\eps=\eps(n)$ below) and separate the proof into two cases: $$2^{- {(2\alpha -1) k}/({{2}\alpha})} \leq \eps \leq  {1}/{6}\quad\text{and}\quad
  2^{-n}\le \eps < 2^{- {(2\alpha -1) k}/({{2}\alpha})}.$$
For the first case, if $k=O(1)$ then the bound we aim for is simply $\widetilde{\Omega}(1/\eps)$, which
  follows trivially from Claim \ref{simpleclaim} (since $k \le \alpha n<n-1$ and the all-0 function is
  a $k$-junta).
Otherwise we~combine~the following reduction with Theorem \ref{thm:main}: any $\eps$-tester for $k$-juntas over $n$-variable functions
  can be used to obtain an $\eps$-tester for $k$-juntas over $(k /\alpha)$-variable functions.
This can be done by adding $n-k/\alpha$ dummy variables to any $(k /\alpha)$-variable function to
  make the number of variables $n$ (as $k \le \alpha n$).
The lower bound then follows from Theorem \ref{thm:main} since $\alpha$ is a constant.
For the second case,
  the lower bound claimed in Theorem~\ref{thm:main-intro} is $\widetilde{\Omega}(1/\eps)$, which follows again from Claim \ref{simpleclaim}.
This concludes the proof of Theorem~\ref{thm:main-intro} given Theorem~\ref{thm:main} and Claim~\ref{simpleclaim}.
\qed

\subsection{Proof of Claim \ref{simpleclaim}}\label{app:simpleclaim}

Let $C$ be a sufficiently large constant.
We prove Claim \ref{simpleclaim} by considering two cases: 
$$
\eps\ge {\frac {C\log n}{2^n}} \quad\text{and}\quad \eps< {\frac {C\log n}{2^n}}.
$$
For the first case of $2^n\eps\ge C\log n$, we use 
 $\calD_1$ to denote the following distribution over $n$-variable Boolean functions: to draw $\bg \sim \calD_1$, independently for each $x \in \zo^n$ the value of $\bg(x)$ is set to~$0$ with probability $1-3\eps$ (recall that $\eps\le 1/5$)
  and $1$ with probability $3\eps$.

We prove the following lemma for the distribution $\calD_1$:

\begin{lemma}
With probability at least $1-o(1)$, $\bg\sim \calD_1$ is $\eps$-far from
  every $(n-1)$-junta.
\end{lemma}
\begin{proof}
Note that every $(n-1)$-junta is such that for some $i \in [n]$, the function does not depend on the $i$-th variable; we refer to such a function as a type-$i$ junta.
An easy lower bound for the distance from a function $g$ to all type-$i$ juntas is the number of
  $g$-bichromatic edges $(x,x^{(i)})$ divided by $2^n$.  When $\bg \sim \calD_1$ each edge $(x,x^{(i)})$ is independently $\bg$-bichromatic with probability 
  $6\eps(1-3\eps)\ge 12\eps/5$ (as $\eps\le 1/5$).
Thus when $2^n\eps\ge C\log n$,
  the expected number of such edges is at least 
$$
2^{n-1}\cdot (12\eps/5)\ge (6/5)\cdot 2^n\eps\ge (6/5)\cdot C\log n.
$$
Using a Chernoff bound, the probability of having fewer than $2^n\eps$ bichromatic edges along direction $i$
  is at most $1/n^2$ when $C$ is sufficiently large. 
The lemma follows from a union bound over $i$. 
\end{proof}

As a result, when $2^n \eps \geq C \log n$, if $\calA$ is a non-adaptive algorithm with the property described in Claim~\ref{simpleclaim},  then $\calA$ must satisfy
$$
\Prx\big[\text{$\calA$ accepts the all-0 function}\hspace{0.03cm}\big]-\Prx_{\bg\sim\calD_1}\big[\text{$\calA$ accepts $\bg$}\hspace{0.03cm}\big]
\ge 2/3-o(1).
$$
But any such non-adaptive algorithm must make $\Omega(1/\eps)$ queries as otherwise with high probability all of its queries to $\bg \sim \calD_1$ would be answered 0, and hence its behavior would be the same as if it were running on the all-0 function.

Finally we work on the case when $1\le 2^n\eps=O(\log n)$.
The proof is the same except that we let $\bg$ be drawn from $\calD_2$, which we define to be the distribution where 
  all entries of $\bg\sim\calD_2$ are $0$ except for exactly $2^n\eps$ of them picked uniformly at random.
The claim follows from the following lemma:
\begin{lemma}
With probability at least $1-o(1)$, $\bg\sim \calD_2$ is $\eps$-far from
  every $(n-1)$-junta.
\end{lemma}
\begin{proof}
This follows from the observation that, with probability $1-o(1)$,
  no two points picked form an edge.
When this happens, we have $2^n\eps$ bichromatic edges along the $i$th direction for all $i$.
\end{proof}
\end{document}